\documentclass[10pt,reqno]{amsart}
\usepackage{latexsym}

\usepackage{amssymb, amscd, amsthm, multirow}
\usepackage[all]{xy}
\usepackage{verbatim}
\usepackage{bbm}
\usepackage{mathrsfs}
\usepackage{indentfirst}
\usepackage{booktabs}
\usepackage{threeparttable}

\newtheorem{theorem}{{Theorem}}
\newtheorem{lemma}{{Lemma}}

\newtheorem{definition}{{Definition}}

\newtheorem{remark}{Remark}

\theoremstyle{remark}
\theoremstyle{claim}

\newtheorem{example}{Example}

\newcommand{\Z}{\mathbb{Z}}
\newcommand{\F}{\mathbb{F}}
\newcommand{\Tr}{\mathrm{Tr}}

\newcommand{\cF}{\mathcal{F}}
\newcommand{\cT}{\mathcal{T}}
\newcommand{\cM}{\mathcal{M}}
\newcommand{\cS}{\mathcal{S}}
\newcommand{\cU}{\mathcal{U}}
\newcommand{\Au}{\mathbf{A1}}
\newcommand{\Av}{\mathbf{A2}}

\newcommand{\w}{\mathbf{w}}

\renewcommand{\Im}{\mathrm{Im}}

\begin{document}
\title[Strictly Optimal Frequency-Hopping Sequence Sets]{Construction of three classes of  Strictly Optimal Frequency-Hopping Sequence Sets}

\author[Yi Ouyang et al.]{Yi Ouyang$^1$, Xianhong Xie$^2$, Honggang Hu$^2$ and Ming Mao$^3$}
\address{$^1$Wu Wen-Tsun Key Laboratory of Mathematics,  School of Mathematical Sciences, University of Science and Technology of China, Hefei, Anhui 230026, China}

\email{yiouyang@ustc.edu.cn}

\address{$^2$University of Science and Technology of China, Key Laboratory of Electromagnetic Space Information, CAS, Hefei, Anhui 230027, China}

\email{hghu2005@ustc.edu.cn}
\email{xianhxie@mail.ustc.edu.cn}

\address{$^3$Beijing Electronic Science and Technology Institute, Beijing 100070, China}

\email{maomingdky@163.com}

\thanks{Research is partially supported by Anhui Initiative in Quantum Information Technologies (Grant No. AHY150200) and NSFC (Grant No. 11571328).}

\begin{abstract}
In this paper, we construct three classes of strictly optimal frequency-hopping sequence (FHS) sets with respect to partial Hamming correlation and family size. The first class is based on a generic construction, the second and third classes are based from the trace map.
\smallskip

\end{abstract}
\maketitle

\section{Introduction}
With advantages such as secure properties, multiple-access, anti-jamming, and anti-fading, frequency-hopping multiple-access (FHMA) is now widely used in modern communication systems such as military communications, bluetooth, sonar echolocation systems and so on \cite{1}. To reduce the multi-access interference, in those systems, the maximum of Hamming out-of-phase autocorrelation and cross correlation of the set of frequency-hopping sequences (FHSs) must be minimized. Thus, it is very interesting to design FHS sets with low Hamming correlation, large size, long period, and small available frequencies simultaneously. In fact, the parameters are not independent with each other, and they are subjected to limitation of some theoretic bounds, for example, the Lempel-Greenberger bound \cite{2}, the Peng-Fan bound \cite{3}, or the coding theory bound \cite{4}. Therefore, it has received a lot of attention about constructing optimal FHSs with respect to the bounds and much progress have been made
(see \cite{5}-\cite{8}, \cite{9}-\cite{15} and the references therein).

The Hamming correlation of FHSs can be divided into three types in general: the periodic Hamming correlation, the aperiodic Hamming correlation, and the partial Hamming correlation. Compared with the first type, the results are relatively little known about the aperiodic and partial ones. However, in the practical application scenarios where the synchronization time is limited or the hardware is complex, the length of a correlation window is usually shorter than the period of the chosen FHSs \cite{17}, i.e., the correlation window length may vary from time to time according to the channel conditions. Consequently, the partial Hamming correlation begin to attract attention. Eun et al. \cite{17} obtained a class of FHSs with optimal partial Hamming correlation from the $m$-sequence and GMW sequences over polynomial residue class ring. In 2012, Zhou et al. \cite{16} derived FHS sets with optimal partial Hamming correlation from trace functions, and generalized the Peng-Fan bounds on the periodic Hamming correlation based on the array structure.
\begin{table}[!htbp]
\setlength{\abovecaptionskip}{0.cm}
\setlength{\belowcaptionskip}{-0.cm}
\caption{Known And Strictly Optimal FHS Sets}
\setlength{\tabcolsep}{0.5mm}{
 \begin{tabular}{|c|c|c|c|c|c|}
     \hline
      \multirow{2}{*}{Length}  & \multirow{1}{*}{Alphabet}  & \multirow{2}{*}{$\cM(\cF;L)$} & \multirow{1}{*}{Family} &\multirow{2}{*}{Constraints} & \multirow{2}{*}{Reference} \\
      & Size& &  Size& & \\
             \hline
       \multirow{2}{*}{$\frac{p^m-1}{r}$} & \multirow{2}{*}{$p^{m-1}$} & \multirow{2}{*}{$\left\lceil\frac{L}{T}\right\rceil$} &  \multirow{2}{*}{$r$}  & \multirow{1}{*}{$\psi(x)$ is identity,}  &\multirow{2}{*}{\cite{16}}\\
       & & & & $r|p-1$, $\gcd(r,m)=1$ &\\
       \hline
       $p^{2m}-1$ & $p^{m}$ & $\left\lceil\frac{L}{p^m+1}\right\rceil$ &  1  &
        & \cite{17}\\
       \hline
      $p(p^m-1)$ & $p^m$ & $\left\lceil\frac{L}{p^m-1}\right\rceil$ &  $p^{m-1}$  &
        & \cite{18}\\
         \hline
       $p^k(p^m-1)$ & $p^m$ & $\left\lceil\frac{L}{p^m-1}\right\rceil$ &  $p^{m-k}$  & $k\mid  m$
        & Theorem~\ref{theorem:2} here\\
       \hline
       \multirow{2}{*}{$\frac{p^m-1}{r}$} & \multirow{2}{*}{$p^{m-1}$} & \multirow{2}{*}{$\left\lceil\frac{L}{T}\right\rceil$} &  \multirow{2}{*}{$r$}  & \multirow{1}{*}{$\psi(x)$,}  &\multirow{2}{*}{ Theorem~\ref{theorem:3} here}\\
       & & & & $r|p-1$, $\gcd(r,m)=1$ &\\
       \hline
\multirow{2}{*}{$\frac{p^{2m}-1}{r}$} &\multirow{2}{*}{$p^m$} & \multirow{2}{*}{$\left\lceil\frac{L}{p^m+1}\right\rceil$ }&\multirow{2}{*}{$r$}  &\multirow{1}{*}{$f(x)$,} & \multirow{2}{*}{Theorem~\ref{theorem:4} here}\\
        & & & & $r|p-1$, $\gcd(r,2m)=1$ &\\
       \hline
   \end{tabular}}\\
   \begin{tablenotes}
        \footnotesize
        \item[1] $T=\frac{q-1}{p-1}$, $q=p^m$ and $\gcd(d,p-1)=1$; $\psi(x)$ is $\F_p$-linear automorphism of $\F_{q}$; $f:\F_{q^2}\rightarrow\F_q$ is difference-balanced function.
      \end{tablenotes}
\end{table}
In 2014, Cai et al. \cite{18} presented FHS sets with optimal partial Hamming correlation from generalized cyclotomy. Later, the authors gave some theoretic bounds of the size of FHS sets and presented a new class of FHSs with optimal partial Hamming correlation in \cite{19}. Very recently, combinatorial constructions of FHSs with optimal partial Hamming correlation have been reported, see \cite{20,21,22}.

The purpose of this paper is to construct three classes of strictly optimal FHS sets with optimal partial Hamming correlation and optimal family sizes. We list the parameters of our construction and related known ones in Table~1, which gives a comparison of our construction and the constructions before us.

Our first construction generates optimal $(p^k(p^m-1),p^{m-k},\left\lceil\frac{L}{p^m-1}\right\rceil;p^k)$ FHS sets with $k\mid m$. It can be viewed as a generalization of the construction \cite{18}. However, under our generalization we can present the more strictly optimal FHS sets that can't be obtained by \cite{18} (see Theorem~\ref{theorem:2}). Although our second and third constructions does not give new parameters, it provides us a large number of choice for $\psi(x)$ and $f(x)$ (see Table~1).

\section{Preliminaries}
Throughout this paper we shall use the following notations.
\begin{itemize}
\item $p$ is an odd prime and $q=p^m>1$ is a $p$-power;
\item $\F_{q}$ is the finite field of $q$ elements, and $\F_{q}^{*}=\F_q-\{0\}$ is the multiplicative group of $\F_{q}$;
\item $\alpha$ is a primitive element of $\F_{q}$;
\item $\Tr_{q^n/q}(x)=\sum\limits^{n-1}_{i=0}x^{q^{i}}$ is the trace map from $\F_{q^n}$ to its subfield $\F_{q}$;
\item For $t\in \Z$, $\langle t\rangle_n\in \{0,\cdots, n-1\}$ denotes the remainder of $t$ by $n$.
\item For all sequences $(x_t)_{t=0}^{n-1}$ indexed by a subscript $t\in \{0,\cdots, n-1\}$, we denote $x_t=x_{\langle t\rangle_n}$ for any $t\in \Z$.
\item For $X$ a finite set, $\# X$ denotes the cardinality of $X$. For the map $f: X\rightarrow Y$ and $y\in Y$, $f^{-1}(y)=\{x\in X:\ f(x)=y\}$ is the set of preimages of $y$.
\end{itemize}

\subsection{Strictly Optimal FHS Sets}
Let $F=\{f_{0},f_{1},\ldots,f_{d'-1}\}$ be an alphabet of $d'$ available frequencies. Let $\cF$ be a  set of $M$ frequency-hopping sequences of the form $(x_i)_{i=0}^{N-1}$ with $x_i\in F$. We also call $\cF$ an $(N,M;d')$-FHS set. For $X=(x_t)_{t=0}^{N-1}$ and $Y=(y_t)_{t=0}^{N-1}$ in $\cF$,  the partial Hamming correlation of $X$ and $Y$ over a window of length $L\in \{1,\cdots,N\}$ starting from $j\in \{0,\cdots, N-1\}$ is
\begin{equation}\label{eq:1}
H_{X,Y}(\tau;j\mid L):=\sum_{t=j}^{j+L-1}h[x_{t},y_{t+\tau}]\quad  (\tau\in \{0,\cdots, N-1\}),
\end{equation}
where $h[a,b]=1$ if $a=b$ and $0$ otherwise (i.e. the Kronecker $\delta$-function). In other words,
\begin{equation}\label{eq:1prime}
H_{X,Y}(\tau;j\mid L)=\#\{t:\ j\leq t\leq j+L-1,\ x_t=y_{t+\tau}\}.
\end{equation}
If $X=Y$ (resp. $X\neq Y$),  $H_{X,X}(\tau;j\mid L)$ (resp. $H_{X,Y}(\tau;j\mid L)$) is called the partial Hamming autocorrelation (resp. cross-correlation) of $X$ (resp. $X$ and $Y$).
Define
 \begin{align}
 & H(X;L):=\max_{\substack{0\leq j<N \\ 1\leq\tau<N}}H_{X,X}(\tau;j\mid L)\quad (X\in \cF), \label{eq:2}\\
 & H(X,Y;L):=\max_{\substack{0\leq j<N \\ 0\leq\tau<N}}H_{X,Y}(\tau;j\mid L)\quad  (X\neq Y\in \cF),\label{eq:3} \\
 & \cM(\cF;L):=\max_{\substack{ X,Y\in \cF \\
 X\neq Y}}\{H(X;L),H(X,Y;L)\label{eq:4}\}.
 \end{align}

In 2009, Niu et al. \cite{23} obtained the following bound on the maximum partial Hamming correlation of FHS sets: for any $(N,M;d')$-FHS set $\cF$, for any window length $1\leq L\leq N$,
\begin{equation} \label{eq:6}
 \cM(\cF;L)\geq\left\lceil \frac{L}{N}\left\lceil\frac{(NM-d')N}{(MN-1)d'}\right\rceil\ \right\rceil,
\end{equation}
and
\begin{equation} \label{eq:7}
\cM(\cF;L)\geq\left\lceil \frac{[2IMN-(I+1)Id']L}{(MN-1)MN}\right\rceil,
\end{equation}
where $I=\left\lfloor\frac{MN}{d'}\right\rfloor$. The partial Hamming correlation bound in ~\eqref{eq:6} is the bound  proved by Eun et al. \cite{17} for the case $M=1$, the Lempel-Greenberger bound by \cite{2} for $L=N$ and $M=1$, and the Peng-Fan bound by \cite{3} for $L=N$.


In 2016, inspired by the idea of Ding et al. \cite{4}, Cai et al.  \cite{18} obtained the following results: for any $(N,M;d')$-FHS set $\cF$,
\begin{equation} \label{eq:8}
 M \leq \min_{2\leq L\leq N}\left\{\left\lfloor \frac{1}{N}\left\lfloor\frac{(L-\cM(\cF;L))d'}
 {L-d'\cM(\cF;L)}\right\rfloor\right\rfloor:\
 L>d'\cM(\cF;L)\right\},
\end{equation}
and
\begin{equation} \label{eq:9}
 M \leq \min_{2\leq L\leq N}\left\{\left\lfloor\frac{d'^{\cM(\cF;L)+1}}{N}
 \right\rfloor:\ L>\cM(\cF;L)\right\}.
\end{equation}

\begin{definition}\label{definition:1} Let $\cF$ be an $(N,M;d')$-FHS set.

(1) $\cF$  is said to be strictly optimal with respect to the partial Hamming correlation if one of the bounds in \eqref{eq:6} or \eqref{eq:7} is achieved for any correlation window length $1\leq L\leq N$.

(2) $\cF$  is said to be strictly optimal  with respect to family size if one of the bounds in \eqref{eq:8} or \eqref{eq:9} is achieved for any correlation window length $2\leq L\leq N$.
\end{definition}

\subsection{Difference-balanced functions}
\begin{definition}\label{definition:2}
A function $f(x): \F_{q^m}\rightarrow \F_{q}$ is called balanced if $\# f^{-1}(x)=q^{m-1}$ for each $x\in \F_q$. It is called difference-balanced if the difference function $f(\delta x)-f(x)$ is balanced for any $\delta\in \F_{q^m}\setminus\{0,1\}$.
\end{definition}
\begin{remark} In the literature (see \cite{32}), balanced and difference-balanced functions are defined over $\F_{q^m}^*$. However, the following is clear.  If assigning $f(0)=0$ to a balanced function $f$ over $\F_{q^m}^*$, one gets a balanced function over $\F_{q^m}$; for a balanced function $f$ over $\F_{q^m}$, then $f-f(0)$ is a balanced function over $\F_{q^m}^*$. If assigning $f(0)=b$ for any $b\in \F_q$ to a difference-balanced function $f$ over $\F_{q^m}^*$, one get a difference-balanced function over $\F_{q^m}$; for a difference-balanced function $f$ on $\F_q$, the restriction of $f$ on $\F_{q^m}^*$ is a difference-balanced function over $\F_{q^m}^*$.
\end{remark}

Pott-Wang~\cite{32} tells us that  a difference-balanced  function $f$ such that $f(0)=0$ is always balanced. Moreover, the following is the list of all known difference-balanced functions from $\F_{q^m}$ to $\F_{q}$ satisfying $f(0)=0$:
\begin{enumerate}
\item[(0)] Functions which are surjective and $\F_q$-linear.

\item[(1)] Functions of the form
\begin{equation}\label{eq:10}
f(x)=\Tr_{q^m/q}(x^d),
\end{equation}
where $d$ is a positive integer prime to $q^m-1$.

\item[(2)] Functions of Helleseth-Gong  type, which was  discovered in \cite{31}.

\item[(3)] Functions of Lin type
\begin{equation}\label{eq:11}
   f(x)=\Tr_{3^m/3}(x+x^s),
   \end{equation}
   where $q=3$, $m=2l+1$ and $s=2\times3^l+1$. The difference balance property of functions of this  type was a conjecture of Lin \cite{25} and proved by Hu et al. \cite{30}.

\item[(4)] Functions which are composites of functions of the previous types (when the composition is legal).
\end{enumerate}
\begin{definition}\label{definition:3}
Let $d$ be an integer prime to $q-1$. A function $f(x)$ from $\F_{q^n}$ onto $\F_{q}$ is called a $d$-form function if
\[f(yx)=y^df(x)\]
for any $y\in \F_{q}$ and $x\in \F_{q^n}$.
\end{definition}

By definition, it is easy to see all the known difference-balanced functions are $d$-form functions: functions of type (0) are $1$-form functions, of type (1)  are $d$-form functions, and  of type (2) and (3)  are $1$-form functions for the case $q=p$.

 \begin{lemma} \label{lemma:balance} Let  $f: \F_{q^2}\rightarrow \F_q$ be a $d$-form difference-balanced function  and $\delta\notin\{0,1\}$,  let $f_\delta(x)=f(\delta x)-f(x)$. Then
 $f^{-1}_\delta(0)=\F_q\cdot a_\delta$ for some $a_\delta\in \F^*_{q^2}$.
 \end{lemma}
\begin{proof}  On one hand $f^{-1}_\delta(0)$ is of order $q$ as $f_\delta$ is balanced. On the other hand,  $f(c)=c^d f(a)$ if $c\in \F_q$, hence $f_\delta(a)=0$ implies that $f_\delta(ca)=0$ for any $c\in \F_q$.
\end{proof}
\begin{remark}  All known difference-balanced functions $f: \F_{q^2}\rightarrow \F_p$ such that $f(0)=0$ are $d$-form functions belonging to one of the following two types: (i) a surjective $\F_q$-linear map; (ii) $(x\mapsto \psi(\Tr_{q^2/q}(x^d)))$ or $(x\mapsto \Tr_{q^2/q}(\psi(x^d)))$ where $\psi$ is an $\F_q$-linear automorphism of $\F_{q^2}$.
\end{remark}

\section{First Construction of Optimal FHS Sets}
\subsection{A generic construction}
From now on, let $q=p^m$.
\begin{definition}\label{definition:4} For  an $(n,M;q)$-FHS set $\cU:=\{U^i=(u^i_t)_{t=0}^{n-1}:0\leq i\leq M-1\}$  and a function $\phi(x)$ over $\F_q$, we say that $(\cU,\phi)$ satisfies $\Au$  if for any given triple $(i,j,\tau)\in [0,M-1]^2\times [0,n-1]-\{(i,i,0)\}$,
\begin{equation}\label{eq:12}
 \phi(u^i_{t+\tau})-\phi(u^j_t)=\text{constant}\ c(i,j,\tau)\in \F^*_{q}\ \text{for all}\ 0\leq t\leq  n-1.
\end{equation}

For a vector $v=(v_0,v_1,\ldots,v_{n'-1})\in \F^{n'}_{q}$ and a function $\psi(x)$ over $\F_{q}$, we say that $(v,\psi)$ satisfies $\Av$ if  for $b\in \F_q$,
\begin{equation}\label{eq:13}
\max_{0\leq t\leq n'-1}\ \# \{t: \psi(v_{t+\tau})-\psi(v_t)=b, 1\leq \tau\leq n'-1\}=1.
\end{equation}
\end{definition}

\begin{theorem} \label{theorem:1} Let $N=nn'$ with $\gcd(n,n')=1$ and $M$ be positive integers. Suppose $\cU$ is an $(n,M;q)$-FHS set such that $(\cU,\phi)$ satisfies $\Au$. Suppose $v\in \F_q^{n'}$ such that $(v,\psi)$ satisfies $\Av$. Then the FHS set $\cS:=\{S^{i}:\ 0\leq i\leq M-1\}$  with $S^{i}=(s^{i}_t)_{t=0}^{N-1}$ defined by
\begin{equation}\label{eq:14}
  s^{i}_t=\phi(u^{i}_t)+\psi(v_t),\quad 0\leq t\leq N-1,
\end{equation}
is an $(N,M;q)$-FHS set and for each correlation window length $L \in \{1,\cdots, N\}$,
\[\cM(\cS;L)\leq \left\lceil\frac{L}{n'}\right\rceil.\]
\end{theorem}

\begin{proof} For $0\leq \tau \leq N-1$, $0\leq i,j\leq M-1$ and $0\leq l\leq N-1$ such that $\tau\neq 0$ if $i=j$, by Eq.~\eqref{eq:1} we have
 \begin{align}
  H_{S^i,S^j}(\tau;l\mid L)&=\sum_{t=l}^{L+l-1}h[s^{i}_t,s^j_{t+\tau}]\notag\\
  &=\sum_{t=l}^{L+l-1}h[\phi(u^{i}_t)+\psi(v_t),\phi(u^{j}_{t+\tau})+\psi(v_{t+\tau})]\notag\\
  &=\sum_{t=l}^{L+l-1}h[\phi(u_{t}^{i})-\phi(u_{ t+\tau}^{j}),\psi(v_{t+\tau})-\psi(v_t)]. \label{eq:15}
\end{align}
Note that $u^i_t:=u^i_{\langle t\rangle_{n}}$, $v_t:=v_{\langle t\rangle_{n'}}$ and $s^i_t:=s^i_{\langle t\rangle_{N}}$ for $t\in \Z$ by our convention.

The triple $(i,j,\tau)$ belongs to two disjoint cases, either
$i=j$ and $\langle\tau\rangle_n=0$ or else. In both cases (the first case is trivial and the second follows from  Eq.~\eqref{eq:12})
\[\phi(u_{t}^{i})-\phi(u_{t+\tau}^{j})=b\in \F_q
\quad \text{for all}\quad t\in \Z, \]
with $b=0$ if and only if $i=j$ and $\langle\tau\rangle_n=0$. Note that if $i=j$, then $\tau\neq 0$, so one must have $\langle \tau\rangle_{n'}\neq 0$ or $b\neq 0$.
Then by \eqref{eq:15}, one has
\begin{equation}\label{eq:16}
  \sum_{t=l+n'z}^{l+n'z+n'-1}h[s_{t}^{i},s_{ t+\tau}^{j}]=\sum_{t=0}^{n'-1}h[b,\psi(v_{t+\tau})-\psi(v_t)]\leq 1
\end{equation}
for any  integers $z$ and $l$. Write the correlation window length $L$ as $L=m_{1}n'+m_{2}$ where $0\leq m_{2}<n'$. From \eqref{eq:15} and \eqref{eq:16}, if $m_2=0$, then
\[
 H_{S^i,S^j}(\tau;l\mid L)=\sum_{t=l}^{l+L-1}h[s_{t}^{i},s_{t+\tau}^{j}]
 =\sum_{z=0}^{m_{1}-1}\sum_{t=l+n'z}^{l+n'z+n'-1}h[s_{t}^{i},s_{ t+\tau}^{j}]
  \leq m_{1},
\]
if $m_2>0$, then
\[
 H_{S^i,S^j}(\tau;l\mid L)
 =\sum_{z=0}^{m_{1}-1}\sum_{t=l+n'z}^{l+n'z+n'-1}
 h[s_{t}^{i},s_{t+\tau}^{j}]+
 \sum_{t=l+m_{1}n'}^{l+m_{1}n'+m_{2}-1}h[s_{t}^{i},s_{t+\tau}^{j}]
 \leq m_{1}+1.
\]
In both cases, we have $H_{S^i,S^j}(\tau;l\mid L)\leq \left\lceil\frac{L}{n'}\right\rceil$.
Hence
\[\cM(\cS;L)\leq \left\lceil\frac{L}{n'}\right\rceil,\]
which completes the proof of this theorem.
\end{proof}
\begin{remark} The special case that $\phi(x)=\psi(x)=x$, $(\cU, \phi)$ and $(v,\psi)$ satisfying $\Au$ and $\Av$ was used in \cite{18} to construct optimal FHS sets.
\end{remark}

\subsection{First class of optimal FHS sets.}

\begin{lemma} \label{lemma:au} Suppose $\F_{p^k}$ is a proper subfield of $\F_{q}$. Let $ \{a_0=1, a_1, \cdots, a_{p^{m-k}-1} \}$ be a set of additive coset representatives of $\F_{p^k}$ to $\F_q$. Let $\sigma$ be a bijection from $\{0,1,\cdots, p^k-1\}$ to $\F_{p^k}$ such that $\sigma(0)=0$.
Construct $p^{m-k}$ sequences $\cU=\{ U^i=(u^i_t)_{t=0}^{p^k-1}: 0\leq i\leq p^{m-k}-1\}$ by
\begin{equation}\label{eq:17}
 u^i_t=a_{i}+\sigma(t).
\end{equation}
Then $\cM(\cU,L)=0$ and $\cU$ is an $(p^k, p^{m-k}-1,0;p^m)$-FHS set. Let $P(x)=x^e +c_{e-1} x^{e-1}+\cdots +c_0\in \F_q[x]$ be a polynomial prime to $x^m-1$, and let $\phi_P(x)=x^{p^e}+c_{e-1} x^{p^{e-1}}+\cdots + c_0 x: \F_q\rightarrow \F_q$. Then for any $(i,j,\tau)\in [0,p^{m-k}-1]^2\times [0,p^k-1]-\{(i,i,0)\}$, for any $0\leq t\leq p^k-1$,
\[\phi_P(u^i_{t+\tau})-\phi_P( u^j_t)=\phi_P(a_i-a_j+\sigma(\tau))\in \F_q^*.\]
Hence $(\cU, \phi_P)$ satisfies $\Au$.
\end{lemma}
\begin{proof} This is because $\phi_P$ is additive and defines a bijection from $\F_q$ to itself, and $a_i-a_j+\sigma(\tau)\neq 0$ if $(i,j,\tau)\neq (i,i,0)$.
\end{proof}

\begin{lemma} \label{lemma:av} Let $\alpha$ be a primitive root of $\F_q$ and $v=(\alpha^t)_{t=0}^{q-2}$. For $\gcd(d,q-1)=1$, set $\psi_d(x)=x^d$, then  $(v, \psi_d)$ satisfies $\Av$.
\end{lemma}
\begin{proof} Easy to check.
\end{proof}

\begin{theorem} \label{theorem:2} Let $(\cU,\phi_P)$ and $(v,\psi_d)$ be given by Lemmas~\ref{lemma:au} and \ref{lemma:av} respectively. Then  the FHS set $\cS$ constructed from $(\cU,\phi_P)$ and $(v,\psi_d)$ in Theorem~\ref{theorem:3} is a $(p^k(p^{m}-1), p^{m-k}; p^m)$-FHS set such that for any correlation window length $1\leq L\leq p^k(p^m-1)$,
 \[ \cM(\cS;L)= \left\lceil\frac{L}{p^m-1}\right\rceil. \]
Moreover, $\cS$ is a strictly optimal FHS set with optimal partial Hamming correlation with respect to the bound in \eqref{eq:7} and with optimal family size with respect to the bound in \eqref{eq:8}.
\end{theorem}
\begin{proof}  By Theorem~\ref{theorem:3}, we know $\cS$ is a $(p^k(p^{m}-1), p^{m-k}; p^m)$-FHS set and $\cM(\cS;L)\leq  \left\lceil\frac{L}{p^m-1}\right\rceil$. We are left to check the equality and the optimality of the partial Hamming correlation and family size with respect to the bounds in Eq.\eqref{eq:7} and \eqref{eq:8}.

For $(N,M,d')=(p^k(p^m-1), p^{m-k}, p^m)$,
\begin{align*}
\left\lceil\frac{(NM-d')N}{(NM-1)d'}\right\rceil
&=\left\lceil\frac{(p^k(p^m-1)p^{m-k}-p^m)p^k(p^m-1)}{(p^k(p^m-1)p^{m-k}-1)p^m}\right\rceil\\
&=\left\lceil\frac{(p^m-2)p^k(p^m-1)}{p^m(p^m-1)-1}\right\rceil=p^k.
\end{align*}
Then Niu et al.'s bound in Eq.\eqref{eq:7} is that for any correlation window length $1\leq L\leq p^k(p^m-1)$, for any $M$ FHS set $\cF$ of length $N$ and alphabet size $d'$,
\[
\cM(\cF;L)\geq\left\lceil \frac{L}{N}\left\lceil\frac{(NM-d')N}{(NM-1)d'}\right\rceil\right\rceil
=\left\lceil\frac{L}{p^m-1}\right\rceil.
\]
Hence $\cM(\cS;L)= \left\lceil\frac{L}{p^m-1}\right\rceil$ and $\cS$ has optimal partial Hamming correlation with respect to  the bound in \eqref{eq:7}.

Take $L=p^m-1$, then $\cM(\cS;L)=1$. The bound \eqref{eq:8} gives
\[M\leq \left\lfloor\frac{p^{2m}}{p^k(p^m-1)}\right\rfloor=\left\lfloor\frac{p^{m-k}p^k(p^m-1)+p^m}{p^k(p^m-1)}\right\rfloor=p^{m-k}.\]
Note that the actual family size of $\cS$ is exactly $p^{m-k}$, hence $\cS$ has an optimal family size with respect to the bound in \eqref{eq:8}.
\end{proof}

\begin{example} Let $p=3$, $m=2$, $d=3$, $k=1$ and $\alpha$ be a primitive element of $\F_{3^2}$ over $\F_3$ satisfying $\alpha^2+\alpha+2=0$. Take $a_i\in\{0,\alpha,2\alpha\}$ with $0\leq i\leq 2$. Suppose $\cU=\{(0,1,2);(\alpha,\alpha+1,\alpha+2);(2\alpha,2\alpha+1,2\alpha+2)\}$ and $v=(\alpha^{t})_{t=0}^{7}$. Set
\[\phi_P(x)=x,\quad \psi(x)=x^3,\]
and $\sigma$ is identity mapping over $\F_3$. Then the sequence set $\cS$ in Theorem~\ref{theorem:4} consists of the following three FHSs of the length 24:
\[ \begin{split}
 S^0=& (1,\alpha^5,\alpha^7,\alpha,0,\alpha,\alpha^2,\alpha^2,0,\alpha^3,\alpha,\alpha^6,\alpha^4,\\
& \alpha^6,\alpha^5,\alpha^5,\alpha^4,\alpha^2,\alpha^6,\alpha^7,1,\alpha^7,\alpha^3,\alpha^3);\\
S^1=&(\alpha^7,0,\alpha^2,\alpha^5,\alpha,\alpha^5,1,1,\alpha,\alpha^4,\alpha^5,\alpha^3,\\
&\alpha^6,\alpha^3,0,0,\alpha^6,1,\alpha^3,\alpha^2,\alpha^7,\alpha^2,\alpha^4,\alpha^4);\\
 S^2=& (\alpha^2,\alpha,1,0,\alpha^5,0,\alpha^7,\alpha^7,\alpha^5,\alpha^6,0,2,\\
&\alpha^3,2,\alpha,\alpha,\alpha^3,\alpha^7,2,1,\alpha^2,1,\alpha^6,\alpha^6).\\
 \end{split} \]
By computer experiments,
 \begin{equation*}
  \cM(\cS;L)=\left\lceil\frac{L}{8}\right\rceil= \begin{cases}
 {1},\ & \textrm{for}\ 1\leq L\leq 8,\\
 {2},\ & \textrm{for}\ 9\leq L\leq 16, \\
 {3},\ & \textrm{for}\ 17\leq L\leq 24. \\
 \end{cases}
\end{equation*}
$\cS$ is strictly optimal with respect to the bound in \eqref{eq:7}, and also has an optimal family size with respect to the bound in \eqref{eq:8}.
\end{example}
\begin{example} Let $p=3$, $m=4$, $d=5$, $k=2$, and $\alpha$ be a root of a primitive polynomial $x^4+x+2$ over $\F_{3}$. Set $\beta=\alpha^{10}\in\F_{3^2}\setminus\F_{3}$, and
\[P(x)=x^2+x-1,\quad \phi_P(x)=x^{9}+x^3-x,\quad \psi(x)=x^7.\]
One know that $\{\alpha,\alpha\beta\}$ can be viewed as a complementary basis of $\F_{3^2}$, and take $a_i\in\{c_0\alpha+c_1\alpha^{11}:c_0,c_1\in \F_3\}$. Suppose $\sigma(a)=\beta^a$ satisfying $\sigma(0)=0$ for $a\in\{1,2,\ldots,8\}$. Then \[\cU=\{(0,\beta,\ldots,\beta^{7},1);(\alpha,\alpha+\beta,\ldots,\alpha+1);\cdots\},\quad v=(\alpha^{t})_{t=0}^{79},\]
and the sequence set $\cS$ in Theorem~\ref{theorem:4} consists of 9 FHSs of the length 720, for $(a_0,a_1,a_2)=(0,\alpha,\alpha^{11})$, three of them are listed below:
\[ \begin{split}
 S^0=& (1,\alpha^{32},\alpha^{65},\alpha^{16},\alpha^{27},\alpha^{64},\alpha^{35},\alpha^{53},\alpha^{33},\alpha^{48},\\
&\alpha^{50},\alpha^{41},\alpha,\alpha^{59},\alpha^{42},\alpha^{32},\alpha^{46},\alpha^{66},\alpha^{25},\alpha^{49},\cdots);\\
S^1=&(\alpha^{51},\alpha^{11},\alpha^{20},\alpha^{66},\alpha,\alpha^{22},\alpha^{41},\alpha^{69},\alpha^{77},\alpha^{70},\\
&\alpha^{19},\alpha^{67},\alpha^{75},\alpha^{10},\alpha^{15},\alpha^{11},\alpha^6,\alpha^{76},\alpha^{14},\alpha^{58},\cdots);\\
 S^2=&(\alpha^3,\alpha^{17},\alpha^{44},\alpha^{63},\alpha^{70},\alpha^{62},\alpha^{76},\alpha^7,\alpha^4,\alpha^{36},\\
&\alpha^{31},\alpha^{16},\alpha^{23},\alpha^9,\alpha^{38},\alpha^{17},\alpha^{64},\alpha^{30},\alpha^{33},\alpha^{19},\cdots).\\
 \end{split} \]
By computer experiments, it can check that $\cS$ is strictly optimal with respect to the bound in \eqref{eq:7}, and also has an optimal family size with respect to the bound in \eqref{eq:8}.
\end{example}

\section{Construction of Optimal FHS Sets via the trace map}
Let $q=p^m$ and $\alpha$ be a primitive root of $\F_q^*$. Let $T=\frac{q-1}{p-1}$. For a nonzero vector $ {\w}=(a_0,a_1,\ldots,a_{k-1})\in \F_{q}^k$, let $V_{\w}=\langle a_0,\cdots, a_{k-1}\rangle_{\F_p}$ be the $\F_p$-subspace of $\F_q$ generated by $\w$ and let $R(\w)=\dim_{\F_p} V_{\w}$.
Define the map $\cT=\cT_{\w}: \F_q\rightarrow \F_p^k$  by
 \begin{equation} \cT(x)=(\Tr_{q/p}(a_{0}x),\Tr_{q/p}(a_{1}x),\ldots,
\Tr_{q/p}(a_{k-1}x)).
 \end{equation}

\begin{lemma} \label{lemma:tp} The map $\cT$
is an $\F_p$-linear map whose kernel $\ker (\cT)=V_{\w}^\bot$ where $V_{\w}^\bot$ is the orthogonal complementary  of $V_{\w}$ via the nondegenerate bilinear  map $\Tr_{q/p}$. In particular,  $\dim_{\F_p}\ker(\cT)= m-R(\w)$ and $\dim_{\F_p} \Im(\cT)=R(\w)$.
\end{lemma}
\begin{proof} An element $x\in \ker (\cT)$ if and only if $\Tr_{q/p}(a_i x)=0$, i.e., $x\bot a_i$ for all $i$, in other words, $x\in V_{\w}^{\bot}$.
\end{proof}
 We construct optimal FHS sequences based on $\w$ with $R(\w)=m-1$ or $m$.

\subsection{Second class of optimal FHS Sets}

\begin{theorem} \label{theorem:3} Fix $\w =(a_0,\cdots, a_{k-1})\in \F_q^k$ such that $R(\w) =m-1$. Let $r\mid p-1$ such that $\gcd(r,m)=1$ and $n'=\frac{q-1}{r}$. Let $v=(\alpha^{rt})_{0\leq t\leq n'-1}$.  Let $d\in \Z$ such that $\gcd(d,q-1)=1$ and $\psi$ be an $\F_p$-linear automorphism of $\F_q$. Define the sequence set
 \[\cS=\cS(\w;r,d, \psi)=\{S^i: 0\leq i\leq r-1\} \]
where
 \begin{equation} \label{eq:19} S^i=(s^i_t)_{0\leq t\leq n'-1}\ \text{and}\
 s^i_t=\cT\circ \psi(\alpha^{d(i+rt)}).\end{equation}
Then $\cS$ is an $(n',r; p^{m-1})$-FHS set and for each correlation window length $1\leq L\leq n'$,
\[\cM(\cS;L)=\left\lceil\frac{L}{T}\right\rceil.\]
Moreover, $\cS$ is a strictly optimal FHS set with optimal partial Hamming correlation with respect to the bound in \eqref{eq:7} and with optimal family size with respect to the bound in \eqref{eq:8}.
\end{theorem}

\begin{proof} As $\gcd(d,q-1)=1$, $\alpha^d$ is still a primitive root of $\F_q^*$. Replacing $\alpha$ by $\alpha^d$, we may assume $d=1$.

The alphabet set $\{s^i_t:\ 0\leq i\leq r-1,\ 0\leq t\leq n'-1\}$ of $\cS$ is nothing but $\Im(\cT)$, hence it is of order $p^{m-1}$ by Lemma~\ref{lemma:tp}.

We are left to compute $\cM(\cS;L)$. Note that $\ker(\cT\circ\psi)=\psi^{-1}\ker(\cT)$ is $1$-dimensional $\F_p$-vector space, pick $0\neq a\in \ker(\cT\circ \psi)$, then $\ker(\cT\circ\psi)=\F_p a$. Fix $i$, $j$ and $\tau$ such that if $i=j$, then $\tau\not\equiv 0\bmod{n'}$. The value $H_{S^i, S^j}(\tau;l\mid T)$ is nothing but  the number of $t\in [l, l+L-1]$ such that $s^i_t=s^j_{t+\tau}$, in other words, $\alpha^{i+rt}-\alpha^{j+r(t+\tau)}=\alpha^{rt} (\alpha^i-\alpha^{j+r\tau})\in \F_p a$. Note that $i-j-r\tau$ is not a multiple of $q-1$, hence $\alpha^i-\alpha^{j+r\tau}\neq 0$. Let $b=(\alpha^i-\alpha^{j+r\tau})^{-1} a$, then
 \[ s^i_t=s^j_{t+\tau}\ \Longleftrightarrow\ \alpha^{rt}\in \F^*_p b. \]
Suppose we have $\alpha^{r t_0}\in \F^*_p b$ for some $t_0$ (otherwise $H_{S^i, S^j}(\tau; l\mid L)=0$). Then
 \[  s^i_t=s^j_{t+\tau}\ \Longleftrightarrow\ \alpha^{r(t-t_0)}\in \F^*_p. \]
Note that $\F^*_p=\{\alpha^{cT}: 0\leq c\leq p-2\}$ and $\gcd(r, T)=\gcd (r,m)=1$, then
  \[  s^i_t=s^j_{t+\tau}\ \Longleftrightarrow\ t- t_0\in T\Z. \]
This means that the number of $t\in [l, l+L-1]$ such that $s^i_t=s^j_{t+\tau}$ is at most $\left\lceil\frac{L}{T}\right\rceil$, i.e., $H_{S^i, S^j}(\tau; l\mid L)\leq \left\lceil\frac{L}{T}\right\rceil$. Hence $\cM(\cS,L)\leq  \left\lceil\frac{L}{T}\right\rceil$.

 On the other hand, the bound Eq.~\eqref{eq:7} gives
\begin{align*}
\cM(\cS;L)&\geq\left\lceil\frac{L}{N}\cdot\frac{[2INM-(I+1)Id']}{(NM-1)M}\right\rceil\\
&=\left\lceil\frac{Lr}{q-1}\cdot\frac{(p-1)(q-2)}{(q-2)r}\right\rceil
=\left\lceil\frac{L}{T}\right\rceil.
\end{align*}
Thus $\cM(\cS,L)= \left\lceil\frac{L}{T}\right\rceil$ and $\cS$ is optimal with respect to the bound in \eqref{eq:7}.

Set $L=\frac{q-1}{r}$, then $\cM(\cS;L)=\frac{p-1}{r}$ and $d'\mathcal{M(S;}L)<L$. According to the bound in \eqref{eq:8}, we have
\[M\leq \left\lfloor\frac{1}{N}\left\lfloor \frac{(L-\cM(\cS;L))d'}{L-d'\cM(\cS;L)}\right\rfloor\right\rfloor=\left\lfloor \frac{rq}{q-1}\right\rfloor=r.\]
Thus the side of $\cS$ is optimal  with respect to the bound in \eqref{eq:8}.
This completes the proof.
\end{proof}
\begin{remark} The special case that $d=1$ and $\psi$ is identity mapping of $\F_q$ was used in \cite{16} to construct optimal FHS sets.
\end{remark}

\begin{example}
Let $p=5$, $m=d=3$, $r=2$ and $(a_0,a_1)=(1,\alpha)$, where $\alpha$ is a primitive element of $F_{5^3}$ over $\F_5$ generated by $\alpha^3+\alpha+1=0$. Suppose $\psi$ is Frobenius automorphism of $\F_{5^3}$.
Then the set $\cS$ of ~\eqref{eq:19} consists of the following two FHSs:
 \[ \begin{split}
S^0&=\left((3,0),(1,3),(1,2),(4,2),(0,0),(1,2),(0,4),(0,4),(1,1),(1,0),(1,4),(2,0),\right.\\
&\left.(3,0),(4,4),(1,4),(4,4),(3,3),(4,2),(3,1),(1,4),(0,1),(3,2),(4,1),(4,2),\cdots\right);\\
S^1&=\left((2,4),(1,4),(2,2),(4,1),(0,0),(2,2),(1,3),(1,3),(3,0),(4,3),(0,1),(3,1),\right.\\
&\left.(2,4),(2,0),(0,1),(2,0),(4,0),(4,1),(1,1),(0,1),(4,2),(0,3),(0,4),(4,1),\cdots\right).\\
 \end{split} \]
By computation,
 \begin{equation*}
  \cM(\cS;L)=\left\lceil\frac{L}{31}\right\rceil= \begin{cases}
 {1},\ & \textrm{for}\ 1\leq L\leq 31,\\
 {2},\ & \textrm{for}\ 31\leq L\leq 62. \\
 \end{cases}
\end{equation*}
$\cS$ is strictly optimal with respect to the bound in \eqref{eq:7}, and also has an optimal family size with respect to the bound in \eqref{eq:8}.
\end{example}
\subsection{Third class of optimal FHS Sets}
\begin{theorem}\label{theorem:4} Fix $\w=(a_0,\cdots, a_{k-1})$ and assume $R(\w)=m$. Let $\theta$ be a primitive element of $\F_{q^2}$ over $\F_q$, $r$ an odd factor of $q-1$, $n'=\frac{q^2-1}{r}$ and $v=(\theta^{rt})_{t=0}^{n'-1}$.
Suppose $f:\F_{q^2}\rightarrow\F_q$ is a $d$-form difference-balanced function.
Define a sequence set
$\cS=:\{S^i=(s^i_t)_{0\leq t\leq n'-1}:0\leq i\leq r-1\}$ by
\begin{equation}\label{eq:26}
 s^i_t=\cT(f(\theta^{i+rt})).
\end{equation}
Then $\cS$ is an $(n',r;q)$-FHS set and for each correlation window length $1\leq L\leq n'$,
 \[ \cM(\cS;L)=\left\lceil\frac{L}{q+1}\right\rceil. \]
Moreover, $\cS$ is a strictly optimal FHS set with optimal partial Hamming correlation with respect to the bound in \eqref{eq:7} and with optimal family size with respect to the bound in \eqref{eq:8}.
\end{theorem}
\begin{proof} By Lemma~\ref{lemma:tp}, we know  the alphabet size of $\cS$ is $p^m$, and the map $\cT$ is injective, hence
\begin{equation}\label{eq:27}
 {s}^i_{t}=s^j_{t+\tau}\ \Longleftrightarrow\  f(\theta^{i+rt})=f(\theta^{j+rt+r\tau}).
\end{equation}
Fix $i,j,\tau$ such that if $i=j$ then $\tau \not\equiv 0\bmod{n'}$. Then $\delta=\theta^{j-i+r\tau}\notin \{0,1\}$. By Lemma~\ref{lemma:balance}, the set $f^{-1}_\delta(0)$ of solutions of $f(\delta x)-f(x)=0$ is $\F_q a$ for some $a\neq 0$. If there exists $\theta^{i+rt_0}\in f^{-1}_\delta(0)$ (otherwise $H_{S^i, S^j}(\tau; l\mid L)=0$), then $\theta^{i+rt_0}\in \F_q^* a$ and
 \[ {s}^i_{t}=s^j_{t+\tau}\ \Longleftrightarrow\ \theta^{i+rt}\in \F_q^* a\ \Longleftrightarrow\  \theta^{r(t-t_0)}\in \F_q^*. \]
Any $c\in \F^*_q$ is of the form $c=\theta^{(q+1)s}$, Since $(r,q+1)=1$, we have
\[ {s}^i_{t}=s^j_{t+\tau}\ \Longleftrightarrow\  t-t_0\in (q+1)\Z. \]
Therefore, we have $H_{S^i, S^j}(\tau; l\mid L)\leq \left\lceil\frac{L}{q+1}\right\rceil$ and
\begin{equation*}
\cM(\cS;L)\leq \left\lceil\frac{L}{q+1}\right\rceil.
\end{equation*}
We now check the strictly optimality of the sequence set $\cS$, from \eqref{eq:7},
\begin{align*}
\cM(\cS;L)&\geq\left\lceil\frac{L}{N}\cdot\frac{[2INM-(I+1)Id']}{(NM-1)M}\right\rceil\\
&=\left\lceil\frac{Lr}{q^2-1}\cdot\frac{(q-1)(q^2-2)}{(q^2-2)r}\right\rceil
=\left\lceil\frac{L}{q+1}\right\rceil.
\end{align*}
Thus $\cM(\cS;L)=\left\lceil\frac{L}{q+1}\right\rceil$.

Taking $L=\frac{q^2-1}{r}$, then $\cM(\cS;L)=\frac{q-1}{r}$ and $d'\cM(\cS;L)<L$. According to the bound in \eqref{eq:8}, we have
\[M\leq \left\lfloor\frac{1}{N}\left\lfloor \frac{(L-\cM(\cS;L))d'}{L-d'\cM(\cS;L)}\right\rfloor\right\rfloor=\left\lfloor \frac{rq^2}{q^2-1}\right\rfloor=r.\]
Actually, the family size of $\cS$ is exactly $r$.
\end{proof}
\begin{remark}
Based on the theory of Galois ring, Eun et al. \cite{17} obtained a class of optimal individual FHS sequence, which have length $q^2-1$, alphabet size $q$, and maximal partial Hamming autocorrelation $\left\lceil\frac{L}{q+1}\right\rceil$ for each window length $1\leq L\leq q^2-1$. Compared with the construction in \cite{17}, our construction gives new parameters when $r>1$, and may be more flexible by the choice of $f(x)$.
 \end{remark}

\begin{example}
Let $p=q=7$, $r=3$, $k=1$ and $a_0=1$. Assume that $\theta$ is a primitive element of $\F_{7^2}$ over $\F_7$ generated by $\theta^2+3\theta-1=0$.  Set
\[f(x)=\Tr_{7^2/7}(x^5).\]
 Then, $\cS$ consists of the following three FHSs of length 16:
 \[ \begin{split}
 &S^0=(2,3,4,1,0,1,3,3,5,4,3,6,0,6,4,4);\\
 &S^1=(6,0,6,4,4,2,3,4,1,0,1,3,3,5,4,3);\\
 &S^2=(3,3,5,4,3,6,0,6,4,4,2,3,4,1,0,1).\\
 \end{split} \]
 By computer experiments,
 \begin{equation*}
  \cM(\cS;L)=\left\lceil\frac{L}{8}\right\rceil= \begin{cases}
 {1},\ & \textrm{for}\ 1\leq L\leq 8,\\
 {2},\ & \textrm{for}\ 9\leq L\leq 16. \\
 \end{cases}
\end{equation*}
Thus, $\cS$ is strictly optimal with respect to the bound in \eqref{eq:7}, and also has an optimal family size with respect to the bound in \eqref{eq:8}.
\end{example}

\section{Conclusion}
In this paper, we constructed three classes of strictly optimal FHS sets with  optimal partial Hamming correlation and optimal family size.


\begin{thebibliography}{99}
\bibitem{1} M. K. Simon, J. K. Omura, R. A. Scholtz, B. K. Levitt: {\em Spread Spectrum Communication Handbook}. New York: McGraw-Hill, 2001.

\bibitem{2} A. Lempel, H. Greenberger: {\em Families of sequences with optimal Hamming-correlation properties}, IEEE Trans. Inf. Theory \textbf{20}(1), 90-94 (1974).

\bibitem{3} D. Peng, P. Fan: {\em Lower bounds on the Hamming auto-and cross correlations of frequency-hopping  sequences}, IEEE Trans. Inf. Theory \textbf{50}(9), 2149-2154 (2004).


\bibitem{4} C. Ding, R. Fuji-Hara, Y. Fujiwara, M. Jimbo, M. Mishima: {\em Sets of frequency hopping sequences: Bounds and optimal constructions}, IEEE Trans. Inf. Theory \textbf{55}(7), 3297-3304 (2009).


\bibitem{5} C. Ding, M. J. Moisio, J. Yuan: {\em Algebraic constructions of optimal frequency-hopping sequences}. IEEE Trans. Inf. Theory \textbf{53}(7), 2606-2610 (2007).

\bibitem{6} J.-H. Chung, Y. K. Han, K. Yang: {\em New classes of optimal frequency-hopping sequences by interleaving techniques}. IEEE Trans. Inf. Theory \textbf{51}(3), 5783-5791 (2009).

\bibitem{7} W. Chu, C. J. Colbourn: {\em Optimal frequency-hopping sequences via cyclotomy}. IEEE Trans. Inf. Theory \textbf{51}(3), 1139-1141 (2005).

\bibitem{8} R. Fuji-Hara, Y. Miao, M. Mishima: {\em Optimal frequency hopping sequences: A combinatorial approach}. IEEE Trans. Inf. Theory \textbf{50}(10), 2408-2420 (2004).

\bibitem{9} G. Solomn: {\em Optimal frequency hopping sequences for multiple access}. in {\em Pro. Symp. Spread Spectr. Commun.,} (1), 33-35 (1973).
\bibitem{10} H. Y. Song, S. W. Golomb: {\em On the nonperiodic cyclic equivalence classes of Reed-Solomon codes}. IEEE Trans. Inf. Theory \textbf{39}(4), 1431-1434 (1993).

\bibitem{11} G. Ge, R. Fuji-Hara, Y. Miao: {\em Further combinatorial constructions for optimal frequency-hopping sequences}. J. Combinat. Theory Ser. A \textbf{113}(8), 1699-1718 (2006).

\bibitem{12} G. Ge, Y. Miao, Z. Yao: {\em Optimal frequency hopping sequences: Auto-and cross-correlation properties}. IEEE Trans. Inf. Theory \textbf{55}(2), 867-879 (2009).

\bibitem{13} X. Zeng, H. Cai, X. Tang, Y. Yang: {A class of optimal frequency hopping sequences with new parameters}. IEEE Trans. Inf. Theory \textbf{58}(7), 4899-4907 (2012).

\bibitem{14} Z. Zhou, X. Tang, D. Peng, U. Parampalli: {\em New constructions for optimal sets of frequency-hopping sequences}. IEEE Trans. Inf. Theory \textbf{57}(6), 3831-3840 (2011).

\bibitem{15} H. G. Hu, G. Gong: {\em New sets of zero or low correlation zone sequences via interleaving techniques}. IEEE Trans. Inf. Theory \textbf{56}(4), 1702-1713 (2010).

\bibitem{16} Z. Zhou, X. Tang,  X. Niu, U. Parampalli: {\em New classes of frequency hopping sequences with optimal partial correlation}. IEEE Trans. Inf. Theory \textbf{58}(1), 453-458 (2012).

\bibitem{17} Y. C. Eun, S. Y. Jin, Y. P. Hong, H. Y. Song: {\em Frequency hopping sequences with optimal partial autocorrelation properties}. IEEE Trans. Inf. Theory \textbf{50}(10), 2438-2442 (2004).

\bibitem{18} H. Cai, Z. Zhou, Y. Yang, X. H. Tang: {\em A new construction of frequency-hopping sequences with optimal partial Hamming correlation}. IEEE Trans. Inf. Theory \textbf{60}(9), 5782-5790 (2014).

\bibitem{19} H. Cai, Y. Yang, Z. C. Zhou, X. H. Tang:{\em Strictly optimal frequency-hopping sequence sets with optimal family sizes}, IEEE Trans. Inf. Theory \textbf{62}(2), 1087-1093 (2016).

\bibitem{20} J. J. Bao: {\em New families of strictly frequency hopping sequence sets}.  Adv. in Math. of Comm. \textbf{12}(2), 387-413 (2018).

\bibitem{21} J. J. Bao, L. J. Ji: {\em Frequency hopping sequences with optimal partial Hamming correlation}. IEEE Trans. Inf. Theory \textbf{62}(6), 3768-3783 (2016).

\bibitem{22} C. L. Fan, H. Cai, X. H. Tang: {\em A combinatorial construction for strictly optimal frequency-hopping sequences}. IEEE Trans. Inf. Theory \textbf{62}(8), 4769-4774 (2016).

\bibitem{23} X. H. Niu, D. Y. Peng, F. Liu: {\em Lower bounds on the periodic partial correlations of frequency hopping sequences with partial low hit zone}. The 4th International Workshop on Signal Design and its Applications in Communications, 84-87 (2009).

\bibitem{24} C. Carlet, C. Ding: {\em Highly nonlinear mappings}. J. Complex., \textbf{20}, 205-244 (2004).

\bibitem{25} A. Lin: {\em From cyclic Hadamard difference sets to perfectly balanced sequences}. Ph.D. dissertation, Dept. Comput. Sci., Univ. Southern California, Los Angeles, CA, USA, 1998.

\bibitem{26} Z. C. Zhou, X. H. Tang, D. H. Wu, Y. Yang: {\em Some new classes of zero-difference balanced functions}. IEEE Trans. Inf. Theory \textbf{58}(1), 139-145 (2012).

\bibitem{27} B. W. Brock: {\em A new construction of circulant $GF(p^2,\Z_p)$}. Discrete Math., \textbf{112} 249-252 (1993).

\bibitem{28} W. de Launey: {\em Circulant $GH(p^2,\Z_p)$ exists for all primes $p$}. Graphs Combin., \textbf{8}, 317-321 (1992).

\bibitem{29} X. H. Niu, C. P. Xing: {\em A construction of optimal frequency hopping sequence set via combination of multiplicative and additive groups of finite fields}. {\em arXiv: 1812.08993v1}.

\bibitem{30} H. G. Hu, S. Shao, G. Gong, T. Helleseth: {\em The proof of Lin's conjecture via the decimation-Hadamard transform}. IEEE Trans. Inf. Theory \textbf{60}(8), 5054-5064 (2014).


\bibitem{31} T. Helleseth and G. Gong, {\em New nonbinary sequences with ideal two-level autocorrelation}. IEEE Trans. Inf. Theory \textbf{48}(11): 2868-2872 (2002).

\bibitem{32} A. Pott and Q. Wang, {\em Some results on difference balanced functions}. in Arithmetic of finite fields, LNCS 9061, Springer, 2015, 111-120.
\end{thebibliography}
\end{document}